\documentclass[11pt]{article}

\usepackage[margin=1in]{geometry}
\usepackage{amsmath,amssymb,amsthm,mathtools}
\usepackage{enumitem}
\usepackage{hyperref}
\usepackage{xcolor}

\hypersetup{
  colorlinks=true,
  linkcolor=blue,
  citecolor=blue,
  urlcolor=blue
}

\newtheorem{theorem}{Theorem}
\newtheorem{lemma}{Lemma}

\newtheorem{remark}{Remark}
\newtheorem{definition}{Definition}

\newcommand{\OPT}{\operatorname{OPT}}

\newcommand{\R}{\mathbb{R}}

\title{A Polynomial-Time $O(\sqrt n)$-Approximation for Undirected Three-Terminal Reachability-Preserving Minimum Edge Cut}

\author{Qi Duan, Carnegie Mellon University}

\date{}

\begin{document}

\maketitle

\begin{abstract}
We study the undirected three-terminal reachability-preserving minimum edge cut problem.
The input is an undirected graph $G=(V,E)$ with nonnegative edge costs, two protected terminals $s_1,s_2$, and a target terminal $t$.
The goal is to remove a minimum-cost edge set so that $t$ is disconnected from the protected terminals while $s_1$ and $s_2$ remain connected.
This problem captures a basic tension between separation and connectivity preservation.
Prior work on connectivity-preserving cuts established polynomial-time solvability for some special cases, such as planar edge-cut instances, and strong hardness for node-cut variants, but a general-graph approximation guarantee for the undirected three-terminal edge-cut version does not appear to have been known.

We give a polynomial-time $O(\sqrt n)$-approximation algorithm, where $n=|V|$.
The main idea is to reformulate the problem as minimizing a rooted terminal cut function over all $s_1$--$s_2$ paths.
For $R=V\setminus\{t\}$ and $A\subseteq R$, define
\[
f(A)=\lambda(t,A)=\min_{U\subseteq R,\ A\subseteq U} c(\delta(U)),
\]
the minimum cost of separating $t$ from every vertex in $A$.
We prove that the RPMEC optimum is
\[
\OPT=\min_{P:s_1\leadsto s_2,\ t\notin V(P)} f(V(P)).
\]
Although $f$ is a monotone submodular function, it has additional rooted cut structure.
We exploit this structure by computing a proportional-fair point in the terminal cut polymatroid
\[
P_f=\{y\ge 0:y(U)\le c(\delta(U))\ \forall U\subseteq R\}.
\]
If $y^\star$ is this point and $\lambda_v=f(\{v\})$, then the weights
\[
a_v=\lambda_v y^\star_v
\]
satisfy
\[
\sqrt{\sum_{v\in A}a_v}
\le
f(A)
\le
\sqrt{|R|}\sqrt{\sum_{v\in A}a_v}
\qquad
\forall A\subseteq R.
\]
Thus $f$ admits a constructive root-linear approximation with factor $\sqrt{|V|-1}$.
The algorithm then finds a shortest $s_1$--$s_2$ path under vertex lengths $a_v$, separates $t$ from that path by a minimum cut, and outputs the connected component containing the path.
The approximation ratio is $\sqrt{|V|-1}$, or $\sqrt{(1+\varepsilon)(|V|-1)}$ with approximate convex optimization.
\end{abstract}

\section{Introduction}
\label{sec:introduction}

Minimum cut is one of the most fundamental problems in graph algorithms and combinatorial optimization.
In its classical form, the goal is to remove a minimum-cost set of edges separating two terminals.
Many applications, however, require more than separation.
A cut may be useful only if some other connectivity relation is preserved.
For example, in communication networks, security-policy analysis, biological interaction networks, and infrastructure protection, one may want to isolate a harmful or undesired target while maintaining communication among protected nodes.
This motivates connectivity-preserving cut problems.

This paper studies the undirected three-terminal reachability-preserving minimum edge cut problem, abbreviated as three-terminal RPMEC.
The input consists of an undirected graph
\[
G=(V,E),
\]
a nonnegative edge-cost function
\[
c:E\to \R_{\ge 0},
\]
two protected terminals $s_1,s_2\in V$, and a target terminal $t\in V$.
The goal is to remove a minimum-cost edge set $F\subseteq E$ such that, in the residual graph $G-F$, the protected terminals $s_1$ and $s_2$ remain connected, while $t$ is disconnected from them.

Equivalently, one may search for a connected source-side region
\[
S\subseteq V\setminus\{t\}
\]
such that
\[
s_1,s_2\in S,
\]
and minimize the boundary cost
\[
c(\delta(S))=\sum_{e\in\delta(S)}c_e.
\]
Here $\delta(S)$ denotes the set of edges with exactly one endpoint in $S$.

The problem has a simple three-terminal description, but it is not a standard minimum cut problem.
An ordinary minimum cut separating $t$ from $\{s_1,s_2\}$ may disconnect $s_1$ from $s_2$.
Conversely, requiring $s_1$ and $s_2$ to remain connected introduces a global connectivity constraint on the source side.
This distinguishes RPMEC from classical minimum cut, multiway cut, and Steiner-type connectivity problems.

The closest classical formulation is the connectivity-preserving minimum cut problem studied by Duan and Xu~\cite{DuanXu2014CPMC}.
They considered both node-cut and edge-cut variants.
Their results show strong inapproximability for node-cut variants and polynomial-time solvability for some planar edge-cut cases.
The general-graph edge-cut case was much less understood.
To the best of our knowledge, no polynomial-time approximation algorithm with a provable performance guarantee was previously known for the general undirected three-terminal edge-cut version.

\subsection{Our contribution}

The main result of this paper is the following.

\begin{theorem}[Main result]
\label{thm:main-intro}
Undirected three-terminal RPMEC admits a polynomial-time $O(\sqrt n)$-approximation, where $n=|V|$.
More precisely, there is a polynomial-time algorithm returning a feasible solution of cost at most
\[
\sqrt{|V|-1}\cdot \OPT.
\]
With approximate numerical optimization, the ratio becomes
\[
\sqrt{(1+\varepsilon)(|V|-1)}
\]
for any fixed $\varepsilon>0$.
\end{theorem}

The proof is based on a terminal-function viewpoint.
Let
\[
R=V\setminus\{t\}.
\]
For any subset $A\subseteq R$, define
\[
f(A)=\lambda(t,A)
=
\min_{\substack{U\subseteq R\\A\subseteq U}} c(\delta(U)).
\]
Thus $f(A)$ is the minimum cost of an edge cut that separates $t$ from every vertex of $A$.
We call $f$ the rooted terminal cut function.

The first key observation is an exact path--mincut identity:
\[
\OPT
=
\min_{P:s_1\leadsto s_2,\ t\notin V(P)} f(V(P)),
\]
where the minimum is over all $s_1$--$s_2$ paths in $G-\{t\}$.
Thus RPMEC can be viewed as choosing an $s_1$--$s_2$ path whose vertex set has minimum rooted terminal cut value.

The second key observation is that $f$ admits a simple root-linear approximation.
Let
\[
P_f=\{y\in\R_{\ge 0}^{R}:y(A)\le f(A)\ \forall A\subseteq R\}
\]
be the polymatroid associated with $f$.
For rooted terminal cut functions, this polymatroid has the graph-cut description
\[
P_f
=
\{y\in\R_{\ge 0}^{R}:y(U)\le c(\delta(U))\ \forall U\subseteq R\}.
\]
We compute the proportional-fair point
\[
y^\star\in
\arg\max_{y\in P_f,\ y_v>0\ \forall v}
\sum_{v\in R}\log y_v.
\]
Let
\[
\lambda_v=f(\{v\})=\lambda(t,v)
\]
be the minimum $t$--$v$ cut value, and define
\[
a_v=\lambda_v y^\star_v.
\]
We prove that, for every $A\subseteq R$,
\[
\sqrt{\sum_{v\in A}a_v}
\le
f(A)
\le
\sqrt{|R|}
\sqrt{\sum_{v\in A}a_v}.
\]
The proof uses only cut feasibility, proportional fairness, and Cauchy's inequality.

The algorithm then finds a shortest $s_1$--$s_2$ path using vertex lengths $a_v$, computes a minimum cut separating $t$ from that path, and outputs the connected component containing the path.

\subsection{Significance}

The result is significant for two reasons.
First, it provides a general-graph approximation guarantee for a connectivity-preserving edge-cut problem for which previous results mainly addressed exact solvability under additional structure, such as planarity.
Second, the method is conceptually different from standard cut approximations.
Instead of reducing RPMEC to multiway cut or directly approximating the connected source side, we approximate the rooted terminal cut function by a vertex-length geometry and then solve a shortest-path problem.

This terminal-function approach may be useful beyond the three-terminal RPMEC setting.
Many connectivity-preserving cut problems contain a connected witness, such as a path, tree, or small subgraph, together with a boundary-cut objective.
If the boundary objective can be represented as a rooted terminal cut function or a related graph-induced submodular function, proportional-fair root-linear approximations may provide new approximation guarantees.

\section{Related Work}
\label{sec:related-work}

\paragraph{Classical cut problems.}
Minimum cut and maximum flow are foundational problems in combinatorial optimization.
In the classical $s$--$t$ minimum cut problem, one seeks a minimum-cost edge set separating two terminals.
By the max-flow min-cut theorem, this problem is polynomial-time solvable and has many efficient algorithms.
Standard references on network flows and cut algorithms include Ahuja, Magnanti, and Orlin~\cite{Ahuja1993NetworkFlows} and Schrijver~\cite{Schrijver2003CombinatorialOptimization}.
Three-terminal RPMEC differs from classical minimum cut because the cut must preserve the connectivity between the two protected terminals.

\paragraph{Connectivity-preserving cuts.}
Connectivity-preserving cut problems impose additional constraints on the residual graph after deletion.
Duan and Xu~\cite{DuanXu2014CPMC} introduced and studied the connectivity-preserving minimum cut problem, including both node-cut and edge-cut versions.
They showed strong hardness and inapproximability results for node-cut variants, while also giving polynomial-time algorithms for important planar edge-cut cases.
The present paper focuses on the undirected three-terminal edge-cut case in general graphs and gives an $O(\sqrt n)$ approximation guarantee.

\paragraph{Multiway cut and multiterminal separation.}
Multiway cut asks for a minimum-cost edge set separating every pair of terminals.
It is NP-hard for three or more terminals~\cite{Dahlhaus1994MultiterminalCuts} and has a rich approximation literature, including classical approximation algorithms by Garg, Vazirani, and Yannakakis~\cite{Garg1996MultiwayCut} and by C{\u a}linescu, Karloff, and Rabani~\cite{Calinescu1998MultiwayCut}.
Although RPMEC also involves three terminals, it is structurally different: $s_1$ and $s_2$ must remain connected, while $t$ must be separated from them.
Thus RPMEC is not a standard multiway cut problem.

\paragraph{Connected source-side cuts.}
The source-side formulation of RPMEC asks for a connected set $S$ containing $s_1$ and $s_2$, excluding $t$, and minimizing $c(\delta(S))$.
This places the problem near connected cut and connected subgraph problems.
However, the objective is a boundary cost, not the cost of constructing the connected subgraph.
Our path--mincut identity separates the connectivity requirement from the boundary objective: connectivity is represented by an $s_1$--$s_2$ path, and the boundary objective is represented by the rooted terminal cut function evaluated on the vertices of that path.

\paragraph{Submodular functions and polymatroids.}
Submodular functions model diminishing returns and are central in combinatorial optimization.
A set function $f:2^R\to\R$ is submodular if
\[
f(A)+f(B)\ge f(A\cup B)+f(A\cap B)
\]
for all $A,B\subseteq R$.
Monotone submodular functions give rise to polymatroids, over which linear optimization can be solved by greedy algorithms.
Standard references include Edmonds~\cite{Edmonds1970SubmodularFunctions}, Lovasz~\cite{Lovasz1983Submodular}, Fujishige~\cite{Fujishige2005Submodular}, and Schrijver~\cite{Schrijver2003CombinatorialOptimization}.
The function $f(A)=\lambda(t,A)$ used in this paper is monotone and submodular, but it also has a special graph-cut representation.

\paragraph{Everywhere approximation of submodular functions.}
Goemans, Harvey, Iwata, and Mirrokni~\cite{Goemans2009ApproximatingSubmodularEverywhere} studied the problem of approximating a submodular function everywhere by a simpler function.
They showed that general nonnegative monotone submodular functions admit root-linear approximations with an $O(\sqrt n\log n)$ factor, while matroid rank functions admit $\Theta(\sqrt n)$-type approximations.
Their framework motivates our use of a root-linear surrogate
\[
g(A)=\sqrt{\sum_{v\in A}a_v}.
\]
However, our proof does not use their general theorem as a black box.
Instead, it exploits the graph-cut structure of rooted terminal cut functions and constructs the weights $a_v$ from a proportional-fair point in the terminal cut polymatroid.
This removes the generic logarithmic factor for this special class.

\section{Preliminaries}
\label{sec:preliminaries}

Let $G=(V,E)$ be an undirected graph.
For a vertex subset $S\subseteq V$, let
\[
\delta(S)=\{uv\in E:u\in S,\ v\in V\setminus S\}
\]
be the cut induced by $S$.
For an edge-cost function $c:E\to\R_{\ge 0}$, define
\[
c(\delta(S))=\sum_{e\in\delta(S)}c_e.
\]

For vertices or vertex sets $X,Y\subseteq V$, we say that $X$ is separated from $Y$ by a cut $\delta(S)$ if $X\subseteq S$ and $Y\cap S=\varnothing$, or vice versa.
Since the graph is undirected, the side of the cut is not oriented.

\begin{definition}[Three-terminal RPMEC]
\label{def:rpmec}
Given an undirected graph $G=(V,E)$, a nonnegative edge-cost function $c:E\to\R_{\ge 0}$, protected terminals $s_1,s_2\in V$, and a target terminal $t\in V$, the three-terminal reachability-preserving minimum edge cut problem asks for a minimum-cost edge set $F\subseteq E$ such that in $G-F$:
\begin{enumerate}[label=(\roman*)]
\item $s_1$ and $s_2$ are connected; and
\item $t$ is disconnected from $s_1$ and $s_2$.
\end{enumerate}
\end{definition}

\begin{definition}[Source-side formulation]
\label{def:source-side}
A set $S\subseteq V$ is a feasible RPMEC source side if
\[
s_1,s_2\in S,
\qquad
 t\notin S,
\]
and the induced subgraph $G[S]$ is connected.
Its cost is
\[
c(\delta(S)).
\]
\end{definition}

Every feasible deletion set $F$ induces a feasible source side by taking the connected component of $G-F$ containing $s_1$ and $s_2$.
Conversely, deleting $\delta(S)$ for any feasible source side $S$ gives a feasible RPMEC solution.
Therefore,
\[
\OPT=
\min\{c(\delta(S)):S\subseteq V\setminus\{t\},\ s_1,s_2\in S,\ G[S]\text{ connected}\}.
\]

\section{An $O(\sqrt n)$-Approximation}
\label{sec:sqrtn-rpmec}

This section proves the main result.
The proof has three steps.
First, we reformulate RPMEC as a path optimization problem with a rooted terminal cut objective.
Second, we approximate this terminal cut objective by a root-linear function.
Third, we solve the resulting shortest-path problem and convert the path back to an RPMEC solution.

\subsection{The rooted terminal cut function}

Let
\[
R=V\setminus\{t\},
\qquad
r=|R|=|V|-1.
\]
For every subset $A\subseteq R$, define
\[
f(A)=\lambda(t,A)
=
\min_{\substack{U\subseteq R\\A\subseteq U}} c(\delta(U)).
\]
The value $f(A)$ is the minimum cost of separating $t$ from every vertex of $A$.
If $A=\varnothing$, then $f(\varnothing)=0$ because the empty set imposes no separation requirement.

\begin{lemma}[Path--mincut identity]
\label{lem:path-mincut}
The optimum value of three-terminal RPMEC satisfies
\[
\OPT
=
\min_{P:s_1\leadsto s_2,\ t\notin V(P)} f(V(P)),
\]
where the minimum is over all $s_1$--$s_2$ paths in $G-\{t\}$.
\end{lemma}

\begin{proof}
Let $S^\star$ be an optimal feasible source side.
Since $G[S^\star]$ is connected and contains $s_1$ and $s_2$, there exists an $s_1$--$s_2$ path $P^\star$ with
\[
V(P^\star)\subseteq S^\star.
\]
The cut $\delta(S^\star)$ separates $t$ from every vertex of $V(P^\star)$.
Therefore,
\[
f(V(P^\star))\le c(\delta(S^\star))=\OPT.
\]
Hence
\[
\min_{P:s_1\leadsto s_2} f(V(P))\le \OPT.
\]

Conversely, let $P$ be any $s_1$--$s_2$ path in $G-\{t\}$.
Let $U\subseteq R$ be an optimal set in the definition of $f(V(P))$, so
\[
V(P)\subseteq U
\]
and
\[
c(\delta(U))=f(V(P)).
\]
Consider the connected component $C$ of the induced subgraph $G[U]$ that contains the path $P$.
Then $C$ contains $s_1$ and $s_2$, excludes $t$, and is connected.
Moreover, no edge of $G[U]$ connects $C$ to $U\setminus C$, because $C$ is a connected component of $G[U]$.
Therefore every edge leaving $C$ also leaves $U$, and hence
\[
\delta(C)\subseteq \delta(U).
\]
It follows that
\[
c(\delta(C))\le c(\delta(U))=f(V(P)).
\]
Thus $C$ is a feasible RPMEC source side of cost at most $f(V(P))$.
Taking the minimum over all paths $P$ gives
\[
\OPT\le \min_{P:s_1\leadsto s_2} f(V(P)).
\]
Combining the two inequalities proves the identity.
\end{proof}

\subsection{Submodularity and the terminal cut polymatroid}

We next record the structural properties of $f$.

\begin{lemma}
\label{lem:f-submodular}
The function $f:2^R\to\R_{\ge 0}$ is normalized, monotone, and submodular.
\end{lemma}

\begin{proof}
The function is normalized because $f(\varnothing)=0$.
It is monotone because if $A\subseteq B$, then every cut separating $t$ from $B$ also separates $t$ from $A$.
Thus
\[
f(A)\le f(B).
\]

To prove submodularity, let $A,B\subseteq R$.
Let $U_A$ and $U_B$ be optimal sets in the definitions of $f(A)$ and $f(B)$, respectively.
Thus
\[
A\subseteq U_A,
\qquad
B\subseteq U_B,
\]
and
\[
f(A)=c(\delta(U_A)),
\qquad
f(B)=c(\delta(U_B)).
\]
The set $U_A\cup U_B$ contains $A\cup B$, so it is feasible for $f(A\cup B)$.
The set $U_A\cap U_B$ contains $A\cap B$, so it is feasible for $f(A\cap B)$.
Since undirected cut capacity is submodular,
\[
c(\delta(U_A))+c(\delta(U_B))
\ge
c(\delta(U_A\cup U_B))+c(\delta(U_A\cap U_B)).
\]
Therefore,
\[
\begin{aligned}
f(A)+f(B)
&=c(\delta(U_A))+c(\delta(U_B))\\
&\ge c(\delta(U_A\cup U_B))+c(\delta(U_A\cap U_B))\\
&\ge f(A\cup B)+f(A\cap B).
\end{aligned}
\]
Thus $f$ is submodular.
\end{proof}

Define the polymatroid associated with $f$:
\[
P_f
=
\{y\in\R_{\ge 0}^{R}:y(A)\le f(A)\ \text{for every }A\subseteq R\}.
\]
For a vector $y\in\R^R$ and a set $A\subseteq R$, we use the notation
\[
y(A)=\sum_{v\in A}y_v.
\]

\begin{lemma}[Cut description of $P_f$]
\label{lem:cut-poly}
The polymatroid $P_f$ has the equivalent description
\[
P_f
=
\{y\in\R_{\ge 0}^{R}:y(U)\le c(\delta(U))\ \text{for every }U\subseteq R\}.
\]
\end{lemma}

\begin{proof}
Suppose first that $y(U)\le c(\delta(U))$ for every $U\subseteq R$.
Let $A\subseteq R$.
For every $U\subseteq R$ with $A\subseteq U$,
\[
y(A)\le y(U)\le c(\delta(U)).
\]
Taking the minimum over all such $U$ gives
\[
y(A)\le f(A),
\]
so $y\in P_f$.

Conversely, suppose $y\in P_f$.
Then for every $U\subseteq R$,
\[
y(U)\le f(U).
\]
Since $U$ itself is feasible in the definition of $f(U)$,
\[
f(U)\le c(\delta(U)).
\]
Therefore
\[
y(U)\le c(\delta(U)).
\]
This proves the equivalence.
\end{proof}

\subsection{The proportional-fair point}

For each vertex $v\in R$, define its singleton terminal connectivity
\[
\lambda_v=f(\{v\})=\lambda(t,v).
\]
In the main proof we assume
\[
\lambda_v>0
\qquad
\text{for all }v\in R.
\]
Zero-connectivity cases are discussed in Remark~\ref{rem:zero-cases}.

Let $y^\star$ be a proportional-fair point of $P_f$, defined by
\[
y^\star
\in
\arg\max_{y\in P_f,\ y_v>0\ \forall v}
\sum_{v\in R}\log y_v.
\]
The logarithm is natural logarithm.
Since the objective is strictly concave and increasing in every coordinate, $y^\star$ is the balanced feasible vector that maximizes the sum of logarithmic utilities.

Now define
\[
a_v=\lambda_v y^\star_v
\qquad
\text{for every }v\in R.
\]
For $A\subseteq R$, define
\[
a(A)=\sum_{v\in A}a_v.
\]
The corresponding root-linear surrogate is
\[
g(A)=\sqrt{a(A)}.
\]

The next theorem is the core of the approximation algorithm.

\begin{theorem}[Root-linear approximation of rooted terminal cut functions]
\label{thm:terminal-root-linear}
For every subset $A\subseteq R$,
\[
\sqrt{a(A)}
\le
f(A)
\le
\sqrt r\,\sqrt{a(A)}.
\]
\end{theorem}

\begin{proof}
We prove the two inequalities separately.

\medskip
\noindent
\textbf{Lower bound.}
We first prove
\[
a(A)\le f(A)^2
\qquad
\text{for every }A\subseteq R.
\]
It is enough to show that for every $U\subseteq R$,
\[
a(U)\le c(\delta(U))^2.
\]
Fix any $U\subseteq R$.
For every $v\in U$, the cut $\delta(U)$ separates $t$ from $v$.
Since $\lambda_v$ is the minimum cost of a cut separating $t$ from $v$, we have
\[
\lambda_v\le c(\delta(U)).
\]
Therefore,
\[
\begin{aligned}
a(U)
&=\sum_{v\in U}\lambda_v y^\star_v\\
&\le c(\delta(U))\sum_{v\in U}y^\star_v\\
&=c(\delta(U))y^\star(U).
\end{aligned}
\]
Since $y^\star\in P_f$, Lemma~\ref{lem:cut-poly} gives
\[
y^\star(U)\le c(\delta(U)).
\]
Hence
\[
a(U)\le c(\delta(U))^2.
\]

Now let $U_A$ be an optimal set in the definition of $f(A)$.
Thus
\[
A\subseteq U_A
\]
and
\[
c(\delta(U_A))=f(A).
\]
Using $a(A)\le a(U_A)$, we get
\[
a(A)\le a(U_A)\le c(\delta(U_A))^2=f(A)^2.
\]
Taking square roots gives
\[
\sqrt{a(A)}\le f(A).
\]

\medskip
\noindent
\textbf{Upper bound.}
We prove
\[
f(A)^2\le r\,a(A).
\]
The proportional-fair optimum satisfies the first-order optimality condition
\[
\sum_{v\in R}\frac{z_v-y^\star_v}{y^\star_v}\le 0
\qquad
\text{for every }z\in P_f.
\]
Equivalently,
\[
\sum_{v\in R}\frac{z_v}{y^\star_v}\le r
\qquad
\text{for every }z\in P_f.
\tag{1}
\label{eq:fairness}
\]

Fix a set $A\subseteq R$.
We need a vector $z^A\in P_f$ supported only on $A$ such that
\[
z^A(A)=f(A).
\]
Such a vector exists by the standard greedy characterization of polymatroids.
For completeness, we give the construction.
Order the elements of $A$ as
\[
A=\{v_1,\ldots,v_k\}.
\]
Let
\[
S_i=\{v_1,\ldots,v_i\}
\]
for $i=1,\ldots,k$, and let $S_0=\varnothing$.
Define
\[
z^A_{v_i}=f(S_i)-f(S_{i-1})
\qquad
\text{for }i=1,\ldots,k,
\]
and define
\[
z^A_v=0
\qquad
\text{for }v\notin A.
\]
By monotonicity, $z^A_v\ge 0$.
By submodularity, this greedy vector belongs to $P_f$.
Moreover,
\[
z^A(A)=\sum_{i=1}^k \bigl(f(S_i)-f(S_{i-1})\bigr)=f(A).
\]
For every $v\in A$,
\[
z^A_v\le f(\{v\})=\lambda_v,
\]
because each greedy marginal is at most the singleton value by submodularity and monotonicity.

Now apply Cauchy's inequality:
\[
\begin{aligned}
f(A)^2
&=\left(\sum_{v\in A}z^A_v\right)^2\\
&=
\left(
\sum_{v\in A}
\sqrt{\lambda_v y^\star_v}
\cdot
\frac{z^A_v}{\sqrt{\lambda_v y^\star_v}}
\right)^2\\
&\le
\left(\sum_{v\in A}\lambda_v y^\star_v\right)
\left(\sum_{v\in A}\frac{(z^A_v)^2}{\lambda_v y^\star_v}\right)\\
&=
a(A)
\left(\sum_{v\in A}\frac{(z^A_v)^2}{\lambda_v y^\star_v}\right).
\end{aligned}
\]
Since
\[
0\le z^A_v\le \lambda_v,
\]
we have
\[
\frac{(z^A_v)^2}{\lambda_v y^\star_v}
\le
\frac{z^A_v}{y^\star_v}.
\]
Therefore,
\[
f(A)^2
\le
a(A)\sum_{v\in A}\frac{z^A_v}{y^\star_v}.
\]
Using \eqref{eq:fairness} with $z=z^A$ gives
\[
\sum_{v\in A}\frac{z^A_v}{y^\star_v}
\le
\sum_{v\in R}\frac{z^A_v}{y^\star_v}
\le r.
\]
Thus
\[
f(A)^2\le r\,a(A).
\]
Taking square roots gives
\[
f(A)\le \sqrt r\,\sqrt{a(A)}.
\]
Combining the lower and upper bounds proves the theorem.
\end{proof}

\subsection{The approximation algorithm}

We now present the algorithm.

\begin{center}
\fbox{
\begin{minipage}{0.92\linewidth}
\textbf{Algorithm: Proportional-fair $O(\sqrt n)$ approximation for undirected three-terminal RPMEC}

\begin{enumerate}[leftmargin=1.5em]
\item Set $R=V\setminus\{t\}$.

\item For every $v\in R$, compute the singleton terminal connectivity
\[
\lambda_v=\lambda(t,v)=f(\{v\}).
\]

\item Compute a proportional-fair point
\[
y^\star\in
\arg\max_{y\in P_f,\ y_v>0\ \forall v}
\sum_{v\in R}\log y_v.
\]

\item Set
\[
a_v=\lambda_v y^\star_v
\qquad
\text{for every }v\in R.
\]

\item Find an $s_1$--$s_2$ path $P$ in $G-\{t\}$ minimizing
\[
a(V(P))=\sum_{v\in V(P)}a_v.
\]
This is a shortest-path computation with nonnegative vertex lengths.

\item Compute a minimum cut separating $t$ from every vertex of $V(P)$.
Equivalently, compute a set $U\subseteq R$ with
\[
V(P)\subseteq U
\]
minimizing $c(\delta(U))$.

\item Let $S$ be the connected component of $G[U]$ that contains $P$.
Output the cut $\delta(S)$, or equivalently the source side $S$.
\end{enumerate}
\end{minipage}
}
\end{center}

The vertex-weighted shortest path in Step 5 can be implemented by the standard vertex-splitting transformation or by assigning vertex costs when relaxing edges.
Since all $a_v$ are nonnegative, Dijkstra's algorithm can be used after a standard transformation.

\begin{theorem}
\label{thm:sqrtn-rpmec}
The algorithm returns a feasible RPMEC solution $S$ satisfying
\[
c(\delta(S))\le \sqrt{|V|-1}\cdot \OPT.
\]
Therefore, undirected three-terminal RPMEC admits a polynomial-time $O(\sqrt n)$-approximation.
\end{theorem}

\begin{proof}
Let $P^\star$ be an $s_1$--$s_2$ path contained in an optimal RPMEC source side.
Such a path exists because the optimal source side is connected and contains $s_1$ and $s_2$.

By the choice of $P$ in Step 5,
\[
a(V(P))\le a(V(P^\star)).
\]
Let $S$ be the connected source side returned by the algorithm.
Since the algorithm computes a minimum cut separating $t$ from $V(P)$, and then possibly restricts to the connected component containing $P$, we have
\[
c(\delta(S))\le f(V(P)).
\]
By Theorem~\ref{thm:terminal-root-linear},
\[
f(V(P))\le \sqrt r\,\sqrt{a(V(P))}.
\]
Therefore,
\[
\begin{aligned}
c(\delta(S))
&\le f(V(P))\\
&\le \sqrt r\,\sqrt{a(V(P))}\\
&\le \sqrt r\,\sqrt{a(V(P^\star))}\\
&\le \sqrt r\,f(V(P^\star)).
\end{aligned}
\]
By Lemma~\ref{lem:path-mincut},
\[
f(V(P^\star))\le \OPT.
\]
Therefore,
\[
c(\delta(S))\le \sqrt r\,\OPT.
\]
Since $r=|V|-1$, the approximation ratio is
\[
\sqrt{|V|-1}.
\]
\end{proof}

\subsection{Polynomial-time computability}

It remains to justify that the proportional-fair point can be computed in polynomial time.

By Lemma~\ref{lem:cut-poly},
\[
P_f
=
\{y\in\R_{\ge 0}^{R}:y(U)\le c(\delta(U))\ \forall U\subseteq R\}.
\]
Given a candidate vector $y$, feasibility can be checked by testing whether there exists $U\subseteq R$ such that
\[
y(U)>c(\delta(U)).
\]
Equivalently, one must compute
\[
\max_{U\subseteq R}\bigl(y(U)-c(\delta(U))\bigr).
\]

This separation problem reduces to one ordinary minimum cut.
Add a new source vertex $q$.
For every $v\in R$, add an edge $(q,v)$ of capacity $y_v$.
Keep all original edges of $G$ with their capacities $c_e$, and take $t$ as the sink.
For a $q$--$t$ cut whose source side is $\{q\}\cup U$, the cut capacity is
\[
y(R\setminus U)+c(\delta(U)).
\]
Since $y(R)$ is constant, minimizing this cut is equivalent to maximizing
\[
y(U)-c(\delta(U)).
\]
Thus if the minimum $q$--$t$ cut has value smaller than $y(R)$, the corresponding set $U$ violates the constraint
\[
y(U)\le c(\delta(U)).
\]
Otherwise, $y$ is feasible.

Therefore $P_f$ has a polynomial-time separation oracle.
The objective
\[
-\sum_{v\in R}\log y_v
\]
is convex on the positive orthant, so the proportional-fair point can be computed to polynomial precision using standard convex optimization with a separation oracle.

For an approximate version, it is enough to compute a feasible vector $y\in P_f$ satisfying
\[
\sum_{v\in R}\frac{z_v}{y_v}\le (1+\varepsilon)r
\qquad
\text{for every }z\in P_f.
\]
Then the same proof as Theorem~\ref{thm:terminal-root-linear} gives
\[
f(A)\le \sqrt{(1+\varepsilon)r}\sqrt{a(A)}
\]
for every $A\subseteq R$, and the final approximation ratio becomes
\[
\sqrt{(1+\varepsilon)(|V|-1)}.
\]

The fairness condition itself can be certified by linear optimization over the polymatroid.
Given
\[
d_v=\frac{1}{y_v},
\]
sort the vertices so that
\[
d_{v_1}\ge d_{v_2}\ge \cdots \ge d_{v_r}.
\]
Let
\[
S_i=\{v_1,\ldots,v_i\},
\qquad
S_0=\varnothing.
\]
The greedy algorithm for polymatroids gives
\[
\max_{z\in P_f} d\cdot z
=
\sum_{i=1}^{r} d_{v_i}\bigl(f(S_i)-f(S_{i-1})\bigr).
\]
Each value $f(S_i)$ is a minimum cut separating $t$ from $S_i$.
Therefore the fairness certificate can be checked using $r$ minimum-cut computations after sorting.

\begin{remark}[Zero-cost and zero-connectivity cases]
\label{rem:zero-cases}
The cleanest proof assumes $\lambda_v>0$ for all $v\in R$.
If $\lambda_v=0$, then $v$ is already separable from $t$ by a zero-cost cut.
Such vertices can be handled by preprocessing.
If $s_1$ and $s_2$ are connected inside a zero-cost-separated component excluding $t$, then the optimum value is $0$.
Otherwise, one may contract zero-cost connected components, restrict the proportional-fair construction to vertices with positive $t$-connectivity, or add an infinitesimal perturbation $\eta>0$ to edge costs and take the limit as $\eta\to 0$.
These standard treatments do not affect the asymptotic approximation guarantee.
\end{remark}

\section{Discussion}
\label{sec:discussion}

The proof can be viewed as a specialized everywhere approximation theorem for rooted terminal cut functions.
General nonnegative monotone submodular functions admit root-linear approximations with an $O(\sqrt n\log n)$ factor~\cite{Goemans2009ApproximatingSubmodularEverywhere}.
In contrast, the rooted terminal cut function
\[
f(A)=\lambda(t,A)
\]
has a graph-cut polymatroid representation:
\[
P_f
=
\{y\ge 0:y(U)\le c(\delta(U))\ \forall U\subseteq V\setminus\{t\}\}.
\]
This additional structure allows us to construct a proportional-fair root-linear surrogate with factor $O(\sqrt n)$.

The approximation ratio is also natural for this type of root-linear embedding.
Consider a star graph in which $t$ is connected independently to $r$ leaves by unit-capacity edges.
Then
\[
f(A)=|A|.
\]
Any root-linear approximation of the form $\sqrt{\sum_{v\in A}a_v}$ can differ from $f(A)$ by a factor of order $\sqrt r$ on the full set.
Conversely, in a graph with one unit-capacity bottleneck edge from $t$ into a component containing all $r$ vertices,
\[
f(A)=1
\qquad
\text{for every nonempty }A.
\]
Again, a $\sqrt r$ loss is unavoidable for simple root-linear surrogates.
Thus improving the RPMEC approximation below $O(\sqrt n)$ would likely require using more structure than a single root-linear embedding of the terminal cut function.

The result should be interpreted specifically for undirected edge cuts.
Directed RPMEC and node-cut RPMEC have different structural properties.
In particular, the proof relies on undirected cut submodularity and on the cut description of the rooted terminal cut polymatroid.
These ingredients do not transfer directly to directed graphs or node-separator formulations.

\section{Conclusion}
\label{sec:conclusion}

We gave a polynomial-time $O(\sqrt n)$-approximation algorithm for undirected three-terminal reachability-preserving minimum edge cut.
The algorithm is based on three ingredients.

First, the problem admits an exact path--mincut formulation:
\[
\OPT
=
\min_{P:s_1\leadsto s_2,\ t\notin V(P)} f(V(P)),
\]
where
\[
f(A)=\lambda(t,A)
\]
is the minimum cost of separating $t$ from every vertex of $A$.

Second, the rooted terminal cut function admits a constructive root-linear approximation.
If $y^\star$ is a proportional-fair point of the terminal cut polymatroid and
\[
a_v=\lambda(t,v)y^\star_v,
\]
then for every $A\subseteq V\setminus\{t\}$,
\[
\sqrt{\sum_{v\in A}a_v}
\le
f(A)
\le
\sqrt{|V|-1}
\sqrt{\sum_{v\in A}a_v}.
\]

Third, the root-linear surrogate turns the path problem into an ordinary shortest-path computation under vertex lengths $a_v$.
Cutting $t$ from the selected path and outputting the connected component containing it gives a feasible RPMEC solution.

The resulting approximation ratio is
\[
\sqrt{|V|-1},
\]
or
\[
\sqrt{(1+\varepsilon)(|V|-1)}
\]
when the proportional-fair point is computed approximately.

Several directions remain open.
The most important is whether undirected three-terminal RPMEC admits a constant-factor approximation on general graphs.
It is also natural to ask whether stronger guarantees are possible for planar graphs, bounded-treewidth graphs, minor-free graphs, or graphs with restricted cut structure.
Finally, extending the proportional-fair terminal-function method to directed or node-cut versions appears challenging and would require new ideas.

\bibliographystyle{plain}
\bibliography{main}

\end{document}